\documentclass[11pt,a4paper,reqno]{article}
\usepackage{cite,amsmath,amsfonts,amsthm}
\usepackage{graphicx}
\usepackage{url}
\usepackage{xcolor}
\allowdisplaybreaks[4] 

\newcommand{\cL}{\mathcal{L}}
\renewcommand{\Re}{\mathop{\mathrm{Re}}}

\theoremstyle{plain}
\newtheorem{thm}{Theorem}
\newtheorem{prop}[thm]{Proposition}
\newtheorem{cor}[thm]{Corollary}

\theoremstyle{definition}
\newtheorem{defn}[thm]{Definition}
\newtheorem{example}[thm]{Example}

\theoremstyle{remark}
\newtheorem{rem}[thm]{Remark}

\begin{document}

\title{\bf Precise option pricing by the COS method--How to choose the truncation range }

\author{
	\sc Gero Junike\footnote{Corresponding author.} and Konstantin Pankrashkin\\[\smallskipamount]
	Carl von Ossietzky Universit\"at\\
	Institut f\"ur Mathematik\\
	26111 Oldenburg, Germany\\[\smallskipamount]
	E-mail: \url{gero.junike@uol.de}, \url{konstantin.pankrashkin@uol.de}
}
 
\date{January 2022\footnote{\color{blue}This is a post-peer-review, pre-copyedit version of the following article:

Junike, G. and Pankrashkin, K. (2022). Precise option pricing by the COS method--How to choose the truncation range, \emph{Applied Mathematics and Computation}, 421, 126935.}}

\maketitle
\begin{abstract}
The Fourier cosine expansion (COS) method is used for pricing European options numerically very fast. To apply the COS method, a truncation range for the density of the log-returns need to be provided. Using Markov’s inequality, we derive a new formula to obtain the truncation range and prove that the range is large enough to ensure convergence of the COS method within a predefined error tolerance. We also show by several examples that the classical approach to determine the truncation range by cumulants may lead to serious mispricing. Usually, the computational time of the COS method is of similar magnitude in both cases.

\bigskip


\bigskip


\end{abstract}

\medskip
\section{Introduction}
In mathematical finance the logarithmic price of a stock is usually
modeled by a random variable $X$. In many financial models the probability density
function $f$ of $X$ exists but its precise structure is unknown. On the other hand the characteristic function $\varphi$ of
$X$ (that is, the Fourier transform of $f$) is often given explicitly, e.g. see models discussed in \cite{fang2009novel}. In this setting, it is necessary to compute integrals of the form
\begin{equation}
\int_\mathbb{R}{v(x)f(x)dx}\label{eq:integral}
\end{equation} 
numerically as fast as possible,
where $v$ is a function describing an insurance contract on a stock, like a \emph{call} or \emph{put
option}, which for example protects the holder against a fall
in the price of the stock. The integral is interpreted as the price of such an insurance
contract. How can we compute the integral without knowing $f$?
A straightforward, efficient and robust method to retrieve the density
of a random variable from its characteristic function
and to compute prices of call or put options is the \emph{COS method} proposed
by Fang and Oosterlee in their seminal work \cite{fang2009novel}. 

It is of utmost importance to price call and put options very fast
because stock price models are typically calibrated to given prices
of liquid call and put options by minimizing the mean-square-error
between model prices and given market prices. During the optimization routine,
model prices of call and put options need to be evaluated very often
for different model parameters.

Under suitable assumptions,
the COS method exhibits exponential convergence
and compares favorably to other Fourier-based pricing techniques, see \cite{fang2009novel}. The COS method is widely applied in mathematical finance, see for instance \cite{bardgett2019inferring,fang2009pricing,fang2011fourier,grzelak2011heston,hirsa2012computational,ruijter2012two,zhang2013efficient}, see \cite{leitao2018data, liu2019neural, liu2019pricing} for an application of the COS method in a data-driven approach. 

The main idea of the COS method is to approximate the density $f$ with infinite support on a finite range $[a,b]$. 
The truncated density is then approximated by a (finite) cosine expansion. The integral (\ref{eq:integral}) can then be approximated highly 
efficiently if $v$ describes the payoff of a put or call option and the characteristic function of $f$ is given in closed-form.

However, it is an open question how to choose the range $[a,b]$ in practice. In this article, we aim to give an answer. More precisely, given some tolerance $\varepsilon>0$, we derive the minimal length of the range $[a,b]$ such that the absolute difference of the approximation by the COS method and the integral in (\ref{eq:integral}) is less than the tolerance.

Fang and Oosterlee \cite{fang2009novel} in Eq. (49), see also Eq. (6.44) in \cite{oosterlee2019mathematical}, proposed some rule of thumb based on cumulants to get an idea how to choose $[a,b]$. In particular, they suggested
\begin{equation}
[a,b]=\begin{cases}
\left[c_{1}\pm 12\sqrt{c_{2}}\right] & ,n_{c}=2\\
\left[c_{1}\pm 10\sqrt{c_{2}+\sqrt{c_{4}}}\right] & ,n_{c}=4\\
\left[c_{1}\pm 10\sqrt{c_{2}+\sqrt{c_{4}+\sqrt{c_{6}}}}\right] & ,n_{c}=6,
\end{cases}\label{eq:cumulants}
\end{equation}
where $c_{1}$, $c_{2}$, $c_{4}$, $c_6$ are the first, second, forth and sixth cumulants of $X$. The parameter $n_c$ may be chosen by the user. If not stated otherwise, we use $n_c=4$. In Section \ref{subsec:Counterexamples}, we provide several examples where the COS method leads to serious mispricing if the truncation range is based on Equation (\ref{eq:cumulants}). Note that there are also Fourier pricing techniques based on wavelets, see \cite{Ortiz2013robust, Ortiz2016Shannon}, which do not relay on an a-priori truncation range.

This article is structured as follows: after reviewing the COS method in Section \ref{sec:Review-COS}, we provide a new proof of convergence of the COS method using elementary tools of Fourier analysis in Section \ref{sec:COS}. Using Markov's inequality, the proof allows us to derive a minimal length for the range $[a,b]$ given a pre-defined error tolerance. Numerical experiments and applications to model calibration can be found in Section \ref{sec:Applications}. Section \ref{Conclusions} concludes.

\section{\label{sec:Review-COS}Review of the COS method.}

Let $f$ be a probability density. The characteristic function $\varphi$ of $f$ is defined by
\begin{equation}
	  \label{charf}
\varphi(u)=\int_{\mathbb{R}}f(x)e^{iux}dx.
\end{equation}
Throughout the article, we assume $f$ to be centered around zero,
that is $\int_{\mathbb{R}}{xf(x)dx}=0$, and we make no additional hypothesis on the support of $f$.
This assumption is mainly
made to keep the notation simple. We
summarize the approximation by the COS method proposed by \cite{fang2009novel}, see also \cite{oosterlee2019mathematical} for detailed explanations.
For $L>0$ let $f_{L}=f1_{[-L,L]}$.  We define the following basis functions
\[
e_{k}^{L}(x):=1_{[-L,L]}(x)\cos\left(k\pi\frac{x+L}{2L}\right),\quad k=0,1,2,...
\]
and the classical cosine-coefficients of $f_{L}$
\[
a_{k}^{L}:=\frac{1}{L}\int_{-L}^{L}f(x)\cos\left(k\pi\frac{x+L}{2L}\right)dx,\quad k=0,1,2,...
\]Let $N\in\mathbb{N}$. The density
$f$ is approximated in three steps:
\begin{equation}
\begin{aligned}
f(x) & \approx f_{L}(x) \approx\sum_{k=0}^{N}{}^{\prime}a_{k}^{L}e_{k}^{L}(x)\approx\sum_{k=0}^{N}{}^{\prime}c_{k}^{L}e_{k}^{L}(x),\label{eq:app4}
\end{aligned}
\end{equation}
where $\sum{}^{\prime}$ indicates that the first summand (with $k=0$)
is weighted by one-half, and we approximate the classical cosine-coefficients $a_{k}^{L}$
by an integral over the whole real line
\begin{equation}
\begin{aligned}
a_{k}^{L} & \approx\frac{1}{L}\int_{\mathbb{R}}f(x)\cos\left(k\pi\frac{x+L}{2L}\right)dx\\ 
 & =\frac{1}{L}\Re\left\{ \varphi\left(\frac{k\pi}{2L}\right)e^{i\frac{k\pi}{2}}\right\} \nonumber \\
 & =:c_{k}^{L},\quad k=0,1,2,....
\end{aligned}
\end{equation}
The coefficients $c_{k}^{L}$ can be computed directly if the characteristic
function of $f$ is given in closed-form. 

Let $v:\mathbb{R}\to\mathbb{R}$ be a (at least locally integrable) function. For $0<M\leq L$ denote 
\begin{equation}
v_{k}^{M}:=\int_{-M}^{M}v(x)\cos\left(k\pi\frac{x+L}{2L}\right)dx,\quad k=0,1,2,...\label{eq:coefficents_v}
\end{equation}
Then \cite{fang2009novel} observed that for $M$ large enough and
replacing $f$ by its approximation (\ref{eq:app4}) it holds
\[
\int_{\mathbb{R}}f(x)v(x)dx\approx\int_{-M}^{M}\sum_{k=0}^{N}{}^{\prime}c_{k}^{L}e_{k}^{L}(x)v(x)dx=\sum_{k=0}^{N}{}^{\prime}c_{k}^{L}v_{k}^{M}
\]
and called the approximation of the integral \emph{COS method}.

Thus the density $f$ is approximated by a sum of cosine functions making use of the characteristic function
to evaluate the cosine coefficients analytically. Working with logarithmic prices, call and put options can be described by
truncated exponential functions and the cosine coefficients $v_{k}^{M}$ of call and put options can be
obtained in explicit form as well. Therefore, option prices can be computed numerically highly efficiently. 

\section{\label{sec:COS} A new framework for the COS method.}

In this section, we revisit the convergence of the COS method. The proof allows us to derive a minimal length for the
finite range $[-L,L]$, given some error tolerance between the integral and its approximation by the COS method. 

Let $ \cL^{1}$ and $ \cL^{2}$ denote the sets of integrable and square integrable real-valued functions on $\mathbb{R}$, and by $\left\langle \cdot, \cdot\right\rangle $ and $\left\Vert\cdot \right\Vert _{2}$ we denote the scalar product and the norm on $\mathcal{L}^{2}$. We denote by $x\lor y$ the maximum value of two real numbers $x$, $y$.

\begin{defn}\label{def:COS admissible}
A  function $f\in\cL^1$ is called \emph{COS-admissible}, if
\begin{equation*}
B(L):=\sum_{k=0}^{\infty}\frac{1}{L}\bigg|\int_{\mathbb{R}\setminus[-L,L]}f(x)\cos\left(k\pi\tfrac{x+L}{2L}\right)dx\bigg|^{2}\rightarrow0
\text{ as } L\rightarrow\infty.
\end{equation*}
\end{defn}

\begin{prop}\label{prop3}
Assume that $f\in\cL^1\cap\cL^2$ with
\begin{equation*}
	\label{provf}
	\int_\mathbb{R} | x f(x)|^2dx<\infty,
\end{equation*}
then
\begin{equation*}
B(L)\leq\frac{2}{3}\frac{\pi^{2}}{L^{2}}\int_{\mathbb{R}\setminus[L,L]}|xf(x)|^{2}dx
\end{equation*} 
and $f$ is COS-admissible.
\end{prop}

\begin{proof} 
We have $B(L)\le 4 (S_L+\widetilde S_L)$ 
with
\begin{align*}
	S_L&:=\sideset{}{'}\sum_{k=0}^\infty \frac{1}{L}\left|\int_L^\infty f(x)\cos\left(k\pi\tfrac{x+L}{2L}\right)dx\right|^{2},\\
	\widetilde S_L&:=\sideset{}{'}\sum_{k=0}^\infty \frac{1}{L}\left|\int_{-\infty}^{-L} f(x)\cos\left(k\pi\tfrac{x+L}{2L}\right)dx\right|^{2},
\end{align*}
and it is sufficient to show $\lim_{L\to\infty} \widetilde S_L=\lim_{L\to\infty} S_L=0$. We will prove it for $S_L$ only, the proof for $\widetilde S_L$ being almost identical.

Let $j\in\mathbb{N}$. Using the classical cosine expansion of $f$ on the interval $[2jL-L,2jL+L]$ and Parseval's identity
we obtain
\begin{equation}
	\label{pars1}
	\begin{aligned}
		\int_{2jL-L}^{2jL+L} |f(x)|^2dx&=\sideset{}{'}\sum_{k=0}^\infty \frac{1}{L}\bigg|\int_{2jL-L}^{2jL+L} f(x)\underbrace{\cos\left(k\pi\tfrac{x-(2jL-L)}{2L}\right)}_{\equiv(-1)^{jk}\cos\left(k\pi\tfrac{x+L}{2L}\right)}dx\bigg|^{2}\\ 
		&=\sideset{}{'}\sum_{k=0}^\infty \frac{1}{L}\left|\int_{2jL-L}^{2jL+L} f(x)\cos\left(k\pi\tfrac{x+L}{2L}\right)dx\right|^{2}.
	\end{aligned}
\end{equation}
Further, using the Cauchy-Schwarz inequality, we estimate
\begin{align*}
	\Big|\int_L^\infty &f(x)\cos\left(k\pi\tfrac{x+L}{2L}\right)dx\Big|^{2}=
	\left|\sum_{j=1}^\infty \tfrac{1}{j}\cdot j\int_{2jL-L}^{2jL+L} f(x)\cos\left(k\pi\tfrac{x+L}{2L}\right)dx\right|^2\\
	&\le\Big( \underbrace{\sum_{j=1}^\infty \tfrac{1}{j^2}}_{=\pi^2/6}\Big) \sum_{j=1}^\infty j^2\left|\int_{2jL-L}^{2jL+L} f(x)\cos\left(k\pi\tfrac{x+L}{2L}\right)dx\right|^2,
\end{align*}
then
\begin{align*}
S_L&\le \dfrac{\pi^2}{6}\sideset{}{'}\sum_{k=0}^\infty \dfrac{1}{L} \sum_{j=1}^\infty j^2\left|\int_{2jL-L}^{2jL+L} f(x)\cos\left(k\pi\tfrac{x+L}{2L}\right)dx\right|^2\\
&=\dfrac{\pi^2}{6} \sum_{j=1}^\infty j^2 \sideset{}{'}\sum_{k=0}^\infty \dfrac{1}{L}\left|\int_{2jL-L}^{2jL+L} f(x)\cos\left(k\pi\tfrac{x+L}{2L}\right)dx\right|^2\\
&\stackrel{\eqref{pars1}}{=} \dfrac{\pi^2}{6} \sum_{j=1}^\infty j^2 \int_{2jL-L}^{2jL+L} |f(x)|^2dx.
\end{align*}
For $x\in [2jL-L,2jL+L]$ one has $j\le \frac{x}{L}$, hence,
\begin{gather*}
j^2 \int_{2jL-L}^{2jL+L} |f(x)|^2dx\le \dfrac{1}{L^2}\int_{2jL-L}^{2jL+L} | x f(x)|^2dx,\\
S_L\le \dfrac{\pi^2}{6} \sum_{j=1}^\infty \dfrac{1}{L^2}\int_{2jL-L}^{2jL+L} | x f(x)|^2dx
=
\dfrac{\pi^2}{6 L^2} \int_{L}^{\infty} | x f(x)|^2dx. 
\end{gather*}
Hence, the assumption \eqref{provf} implies $\lim_{L\to\infty}S_L=0$.
\end{proof}

\begin{cor}
\label{lem:f bounded_square_integrable}Let the density $f$ be bounded,
with finite first and second moments, then $f$ is COS-admissible.
\end{cor}

Corollary \ref{lem:f bounded_square_integrable} already shows that the class of COS-admissible densities is very large. Next, we provide further sufficient conditions for COS-admissibility. In particular, we show that the densities of the stable distributions for stability parameter $\alpha \in(\frac12,2]$, including the Normal and the Cauchy distributions, and the density of the Pareto distribution are COS-admissible.

\begin{cor}\label{cor4}
	Let $f\in\cL^1$ such that its characteristic function $\varphi$ defined in \eqref{charf} has a
	weak derivative $\varphi'$ satisfying $|\varphi|^2+|\varphi'|^2\in\cL^1$.
	Then $f$ is COS-admissible.
\end{cor}

\begin{proof} 
As known from the Fourier analysis of tempered distributions, e.g. \cite[Chap. VII]{horm},
	the function $\varphi'$
is the  Fourier transform of $x\mapsto ixf(x)$,
and due to Plancherel theorem we have
\[
\int_{\mathbb{R}} \big(|f(x)|^2+|xf(x)|^2\big)dx=\dfrac{1}{2\pi} \int_{\mathbb{R}}\big(|\varphi(u)|^2+|\varphi'(u)|^2 \big)du<\infty.
\]
Hence, the assumptions of Proposition~\ref{prop3} are satisfied.
\end{proof}

\begin{example} 
The densities of the stable distributions
whose characteristic functions are of the form
\begin{gather*}
\varphi(u)=\exp \big[i\mu u - |c u|^\alpha
\big (1-i\beta \Phi_\alpha(u) \mathop{\mathrm{sgn}} u\big )\big],\\
\Phi_\alpha(u)=\begin{cases}
	\tan\frac{\pi\alpha}{2},& \alpha\ne 1,\\
	-\frac{2}{\pi} \log|c u|, & \alpha=1,
	\end{cases}
\end{gather*}
for parameters $\alpha \in(\frac12,2]$, $\beta\in[-1,1]$, $c>0$ and $\mu \in\mathbb{R}$
are COS-admissible by Corollary \ref{cor4}.
Recall that this class includes the Normal and the Cauchy distributions.
\end{example}

\begin{example}
The density of the Pareto distribution with scale $\beta>0$ and shape
$\alpha>0$ can be described by $f(x)=\alpha\beta^{\alpha}x^{-(\alpha+1)}$,
for $x\geq\beta$, and is COS-admissible. To see this, let $B(L)$
as in Definition \ref{def:COS admissible}. It holds by integration by parts and using $\sum_{k=1}^{\infty}\frac{1}{k^{2}}<\infty$ and $\sin(k\pi)=0$, $k=1,2,...$,
that
\begin{align*}
B(L)= & \frac{1}{L}\left|\int_{L}^{\infty}f(x)dx\right|^{2}\\
 & +\sum_{k=1}^{\infty}\frac{1}{L}\left|-\frac{2L}{k\pi}\int_{L}^{\infty}f^{\prime}(x)\sin\left(k\pi\frac{x+L}{2L}\right)dx\right|^{2}\\
\leq & \frac{\beta^{2\alpha}}{L^{2\alpha+1}}+\frac{4\alpha^{2}\beta^{2\alpha}}{\pi^{2}L^{2\alpha+1}}\sum_{k=1}^{\infty}\frac{1}{k^{2}}\to0,\quad L\to\infty.
\end{align*}
\end{example}

Now we discuss the use of COS-admissible functions for the approximation of some integrals arising in mathematical finance. The next theorem shows that in particular a density with infinite support can be approximated by a cosine expansion.

\begin{thm}
\label{thm8}
Assume $f\in \cL^1\cap\cL^2$ to be COS-admissible, then
\[
\lim_{L\to\infty}\limsup_{N\to\infty} \Big\|f-\sideset{}{'}\sum_{k=0}^{N}c_{k}^Le_{k}^L\Big\|_2=0.
\]
\end{thm}

\begin{proof}
Let $f_L$, $a_{k}^{L}$, $e_{k}^{L}$, $c_{k}^{L}$, $v_{k}^{M}$ as in Section \ref{sec:Review-COS}. Let $L>0$. Recall that $\langle e_{0}^L,e_{0}^L\rangle=2L$ and $\langle e_{k}^L,e_{l}^L\rangle=L \delta_{k,l}$ for $(k,l)\ne(0,0)$.
Consider the cosine coefficients of the tails of $f$, defined by
\[
\tilde{c}_{k}^L=\frac{1}{L}\int_{\mathbb{R}\setminus[-L,L]}f(x)\cos\left(k\pi\tfrac{x+L}{2L}\right)dx.
\]
Then it holds $c_{k}^L=a_{k}^L+\tilde{c}_{k}^L$, and it follows 
\[
\bigg\| f-\sideset{}{'}\sum_{k=0}^{N}c_{k}^Le_{k}^L\bigg\|_{2}
\leq
\underbrace{\big\Vert f-f_{L}\big\Vert_{2}}_{=:A_{1}(L)}
 +
 \underbrace{\Big\Vert f_{L}-\sideset{}{'}\sum_{k=0}^{N}a_{k}^Le_{k}^L\Big\Vert_{2}}_{=:A_{2}(L,N)}
 +
 \underbrace{\Big\Vert \sideset{}{'}\sum_{k=0}^{N} \tilde{c}_{k}^Le_{k}^L\Big\Vert_2}_{=:A_3(L,N)}.
\]
Due to $f\in\cL^2$ it holds $\lim_{L\to\infty}A_{1}(L)=0$. For each fixed $L$
one has $\lim_{N\to\infty}A_{2}(L,N)=0$ as $f$ is square integrable on each $[-L,L]$.
Further, 
\begin{align*}
A_3(L,N)^2&=\Big\langle \sideset{}{'}\sum_{k=0}^{N}\tilde{c}_{k}^Le_{k}^L,\sideset{}{'}\sum_{k=0}^{N}\tilde{c}_{k}^Le_{k}^L\Big\rangle=L\sideset{}{'}\sum_{k=0}^N |\Tilde c_k^L|^2\le L \sum_{k=0}^\infty |\Tilde c_k^L|^2\\
&= \sum_{k=0}^\infty \dfrac{1}{L} \bigg|\int_{\mathbb{R}\setminus[-L,L]}f(x)\cos\left(k\pi\tfrac{x+L}{2L}\right)dx\bigg|^{2}=B(L),
\end{align*}
and $\lim_{L\to\infty}B(L)=0$ as $f$ is COS-admissible.

Take any $\varepsilon>0$ and choose $L_0$ such that
$A_1(L)<\frac{\varepsilon}{3}$ and $B(L)<\big(\frac{\varepsilon}{3}\big)^2$ 
for all $L>L_0$, then
$A_3(L,N)<\frac{\varepsilon}{3}$ for all $L>L_0$ and all $N$.
For any $L>L_0$, choose $N_L$ sufficiently large  to have $A_2(L,N)<\frac{\varepsilon}{3}$
for all $N>N_L$, then
\[
\Big\Vert f-\sideset{}{'}\sum_{k=0}^{N}c_{k}^Le_{k}^L\Big\Vert _{2}<\varepsilon
\text{ for all $L>L_0$ and all $N>N_L$,}   
\]
which finishes the proof.\end{proof}

The next corollary shows that the integral \eqref{eq:integral} can be computed very efficiently if the
characteristic function of $f$ and the classical cosine coefficients of $v$ are
available in analytical form: this justifies the use of the COS method under some mild technical 
assumptions on the density $f$ and the function $v$. 
The most important example in the applications
for $v$ is the call option  $x\mapsto \max(e^{x}-K,0)$ and the put option $x\mapsto \max(K-e^{x},0)$  for some $K\geq0$. 
The coefficients \eqref{eq:coefficents_v} can be obtained analytically for call and put options. 

\begin{cor}\label{cor:Integral}
Let $f\in\cL^{1}\cap\cL^{2}$ be COS-admissible and
$v:\mathbb{R}\to\mathbb{R}$ be locally in $\cL^2$, that is,
\[
 \text{$v_{M}:=v1_{[-M,M]}\in\cL^2$ for any $M>0$.}
\]
Assume that $vf\in\cL^{1}$, then the integral of the
product of $f$ and $v$ can be approximated by a finite sum as follows.

Let $\varepsilon>0$ and let $M>0$ and $\xi>0$ be such that
\begin{equation}
\int_{\mathbb{R}\setminus[-M,M]}\big|v(x)f(x)\big|dx\leq\frac{\varepsilon}{2},\quad
\|v_M\|_2\le \xi.\label{eq:vf<eps/2}
\end{equation}
Let $L\geq M$ such that
\begin{equation}
\big\Vert f-f_{L}\big\Vert_{2}\leq\ \frac{\varepsilon}{6\xi}\label{eq:L1}
\end{equation}
and
\begin{equation}
B(L)\leq\left(\frac{\varepsilon}{6\xi}\right)^{2}.\label{eq:L2}
\end{equation}
Choose $N_{L}$ large enough so that
\[
\left\Vert f_{L}-\sum_{k=0}^{N}{}^{\prime}a_{k}^{L}e_{k}^{L}\right\Vert _{2}\leq\frac{\varepsilon}{6\xi},\quad N\geq N_{L}.
\]
Then it holds for all $N\geq N_{L}$
\begin{equation}
\left|\int_{\mathbb{R}}v(x)f(x)dx-\sum_{k=0}^{N}{}^{\prime}c_{k}^{L}v_{k}^{M}\right|\leq\varepsilon.
\end{equation}
\end{cor}

\begin{proof}
Let $A_1(L)$ and $A_2(L,N)$ as in the proof of Theorem \ref{thm8}.
Due to $v_{k}^{M}=\langle v_{M},e_{k}^L\rangle$,  for all $N\ge N_L$ and by Theorem \ref{thm8} one obtains 
\begin{align*}
\Big|\int_{\mathbb{R}}v(x)&f(x)dx  -\sideset{}{'}\sum_{k=0}^{N}c_{k}^Lv_{k}^{M}\Big|\\
&=\Big|
\int_{\mathbb{R}\setminus[-M,M]}v(x)f(x) dx+
\langle v_M,f\rangle 
-\sideset{}{'}\sum_{k=0}^{N}c_{k}^L\langle v_{M},e_{k}^L\rangle 
\Big|\\
&\leq\Big|\int_{\mathbb{R}\setminus[-M,M]}v(x)f(x)dx\Big|+\Big|\Big\langle v_{M},f-\sideset{}{'}\sum_{k=0}^{N}c_{k}^Le_{k}^L\Big\rangle\Big|\\
 & \leq\int_{\mathbb{R}\setminus[-M,M]}|v(x)f(x)|dx+\| v_{M}\|_2\,\Big\Vert f-\sideset{}{'}\sum_{k=0}^{N}c_{k}^Le_{k}^L\Big\Vert_{2}\\
 & <\frac{\varepsilon}{2}+ \xi\left( A_1(L) + A_2(L,N) + \sqrt{B(L)}  \right)
 \\
  & \leq\frac{\varepsilon}{2}+ \xi\left( \frac{\varepsilon}{6\xi} + \frac{\varepsilon}{6\xi} + \frac{\varepsilon}{6\xi}  \right)=\varepsilon.
 \\
\end{align*}
\end{proof}

In the next corollary, we apply Corollary \ref{cor:Integral} to
bounded functions $v$ and, given some $\varepsilon>0$, obtain explicit formulae for $M$ and $L$ such that the Inequalities (\ref{eq:vf<eps/2}), (\ref{eq:L1}) and (\ref{eq:L2}) hold. For example, put options are bounded. To find
a lower bound for $M$, we need to estimate the \emph{tail sum} of
$f$. We apply Markov's inequality to estimate the tail sum using
the $n^{th}-$moment of $f$. 

If $f$ has semi-heavy tails, i.e. the
tails of $f$ decay exponentially, all moments of $f$ exist and
can be obtained by differentiating the characteristic function $n$
times. In a financial
context, log-returns are often modeled by semi-heavy tails or lighter
distributions: for instance the distribution of log-returns in the Heston model is between the exponential
and the Gaussian distribution, see \cite{dragulescu2002probability}. See \cite{schoutens2003levy} for an
overview of Lévy models with semi-heavy tails. In particular, the
Lévy process CGMY or the generalized hyperbolic processes, see \cite{albin2009asymptotic},
have semi-heavy tails.

We can find a suitable value of $L$ with the help of the bound for $B(L)$ in  Proposition \ref{prop3}.
For some Lévy models, the density $f$ is given in closed-form, see \cite{schoutens2003levy}, and
$B(L)$ can be estimated directly. 
In the next corollary, we propose to approximate the tails of 
$f$ by a density $\lambda$, which is given in closed-form.
We suggest to use the Laplace density, which decays exponentially just
like a density with semi-heavy tails. The (central) Laplace density
has only one free parameter describing the variance. One can use a
moment-matching method to calibrate the Laplace density $\lambda$,
setting the variances corresponding to $f$ and $\lambda$ equal.

The Laplace density with mean zero and variance $\sigma^{2}$ ($\sigma>0$) is given
by
\begin{equation}
\lambda_{\sigma}(x)=\frac{1}{\sqrt{2}\sigma}e^{-\sqrt{2}\frac{|x|}{\sigma}},\quad x\in\mathbb{R}.\label{eq:Laplace density}
\end{equation}

\begin{cor}[COS method, Markov range]
\label{cor:bounded v}
Let both the density $f$ and the function $v$ be bounded, with $|v(x)|\le K$ for all $x\in\mathbb{R}$ and some $K>0$.
For some even natural number $n\geq2$ assume the $n^{th}-$moment of $f$ exists and denote it by $\mu_{n}$, i.e.
\begin{equation}
\mu_{n}=\int_\mathbb{R}{x^nf(x)dx}<\infty.
\end{equation}
Let $\varepsilon>0$ and
\begin{equation}
M=\sqrt[n]{\frac{2K\mu_{n}}{\varepsilon}}.\label{eq:M}
\end{equation}
Assume there is some $\sigma>0$ such that 
\begin{equation}
f(x)\leq\lambda_{\sigma}(x),\quad|x|\geq M.\label{eq:upper_bound_f}
\end{equation}
Set
\begin{align}
L= & M\lor\left(-\frac{\sigma}{2\sqrt{2}}\log\left(\frac{\sqrt{2}\sigma\varepsilon^{2}}{72MK^{2}}\frac{12}{\pi^{2}}\left(\frac{\sigma^{2}}{M^{2}}+\frac{2\sqrt{2}\sigma}{M}+4\right)^{-1}\right)\right)\nonumber \\
 & \lor-\frac{\sigma}{2\sqrt{2}}\log\left(\frac{2\sqrt{2}\sigma\varepsilon^{2}}{72MK^{2}}\right).\label{eq:LL}
\end{align}
Choose $N_{L}$ large enough. Then it holds for all $N\geq N_{L}$, that
\begin{equation}
\left|\int_{\mathbb{R}}v(x)f(x)dx-\sum_{k=0}^{N}{}^{\prime}c_{k}^{L}v_{k}^{M}\right|\leq\varepsilon.\label{eq:vf-sum ck vk}
\end{equation}
\end{cor}

\begin{proof}
It holds by Markov's inequality and the definition of $M$
\[
\int_{\mathbb{R}\setminus[-M,M]}\left|v(x)f(x)\right|dx\leq K\int_{\mathbb{R}\setminus[-M,M]}f(x)dx\leq K\frac{\mu_{n}}{M^{n}}\leq\frac{\varepsilon}{2}.
\]
Hence Equation (\ref{eq:vf<eps/2}) is satisfied. Next we show Inequality
(\ref{eq:L1}) holds. Let $\xi=\|v_M\|_2$ as in Corollary \ref{cor:Integral}.
As $v$ is bounded by $K$ it holds $\xi^{2}\leq2MK^{2}$. Using the
upper bound (\ref{eq:upper_bound_f}) and the expressions (\ref{eq:Laplace density})
it follows 
\begin{align*}
\left\Vert f-f_{L}\right\Vert _{2}^{2} & =\int_{\mathbb{R}\setminus[-L,L]}f^{2}(x)dx\\
 & \leq\frac{1}{2\sigma^2}\int_{\mathbb{R}\setminus[-L,L]}e^{-2\sqrt{2}\frac{|x|}{\sigma}}dx\\
 & =\frac{1}{2\sqrt{2}\sigma}e^{-2\sqrt{2}\frac{L}{\sigma}}\leq\frac{\varepsilon^{2}}{72MK^{2}}\leq\left(\frac{\varepsilon}{6\xi}\right)^{2}.
\end{align*}
The second last inequality follows by the definition of $L$ and the last inequality holds as $\xi^{2}\leq2MK^{2}$. Hence Inequality
(\ref{eq:L1}) is satisfied. To show that also Inequality (\ref{eq:L2})
is respected, we use Proposition \ref{prop3}. It holds

\begin{align*}
B(L) & \leq\frac{2}{3}\frac{\pi^{2}}{L^{2}}\int_{\mathbb{R}\setminus[-L,L]}|xf(x)|^{2}dx\\
 & \leq\frac{4}{3}\frac{\pi^{2}}{L^{2}}\int_{L}^{\infty}x^{2}\lambda_{\sigma}^{2}(x)dx\\
 & =\frac{\pi^{2}}{12\sqrt{2}\sigma}\left(\frac{\sigma^{2}}{L^{2}}+\frac{2\sqrt{2}\sigma}{L}+4\right)e^{-2\sqrt{2}\frac{L}{\sigma}}\\
 & \leq\frac{\pi^{2}}{12\sqrt{2}\sigma}\left(\frac{\sigma^{2}}{M^{2}}+\frac{2\sqrt{2}\sigma}{M}+4\right)e^{-2\sqrt{2}\frac{L}{\sigma}}\\
 & \leq\frac{\varepsilon^{2}}{72MK^{2}}\leq\left(\frac{\varepsilon}{6\xi}\right)^{2},
\end{align*}
by the definition of $L$. By Corollary \ref{cor:Integral}, Inequality
(\ref{eq:vf-sum ck vk}) is obtained.
\end{proof}

\begin{rem}
The original proof of the convergence of the COS method given in \cite{fang2009novel} is somewhat more restrictive compared to the results stated in Corollary \ref{cor:Integral} because it is assumed that the sum over the cosine coefficients of the payoff function is finite, i.e.,
\begin{equation}
\sum_{k=N}^{\infty}\left|v_{k}^{M}\right|<\infty,\label{harmonicSum}
\end{equation}
see Lemma 4.1 in \cite{fang2009novel}. In particular, the cosine coefficients of a put option decay as $\frac{1}{k}$, see Appendix \ref{sec:Interval-cumulants}, and do not satisfy Equation (\ref{harmonicSum}). On the other hand, according to Corollary \ref{cor:bounded v}, put options can be priced very well by the COS method.
\end{rem}

\begin{rem}
For small $\varepsilon$, the parameter $M$ is of order $\varepsilon^{-1/n}$, while the other terms in \eqref{eq:LL} are of order $\log\varepsilon$. Hence, $L=M$ for $\varepsilon$ small enough, i.e. the formula for $L$
simplifies considerably. Indeed, for the numerical experiments reported
in Table \ref{tab:Advanced-stock-price} we observed $L=M$. 
\end{rem}

\begin{rem}
The smoother $f$, the smaller $N_{L}$ may be chosen to obtain a
good approximation, see \cite{fang2009novel} and references therein.
\end{rem}

\begin{rem}
The COS method has also been applied to price exotic options like Bermudan, American or discretely monitored barrier and Asian options, see 
\cite{fang2009pricing,fang2011fourier,zhang2013efficient}. One of the central idea when pricing these exotic options is the approximation of the density of the log returns by a cosine expansion as in Theorem \ref{thm8}. But some care is necessary when applying Corollary \ref{cor:Integral} to obtain a truncation range, because the truncation range also depends on the payoff function itself. 

The COS method has been extended to the multidimensional case, see \cite{ruijter2012two}, using a heuristic truncation range similar to Equation (\ref{eq:cumulants}). In a future research, we would like to generalize our results, in particular generalize Definition \ref{def:COS admissible}, Theorem \ref{thm8} and Corollary \ref{cor:bounded v}, and derive a truncation range in a multidimensional setting.
\end{rem}

\section{Applications}\label{sec:Applications}

In this section, we apply Corollary \ref{cor:bounded v} to some stock price models.
We use the following setting: let $(\Omega,\mathcal{F},Q)$ be a probability
space. $Q$ is a risk-neutral measure. Let $S_{T}$ be a positive random
variable describing the price of the stock at time $T>0$. The price of
the stock today is denoted by $S_{0}$. We assume there is a bank
account paying continuous compounded interest $r\in\mathbb{R}$. Let
\begin{equation}
X:=\log(S_{T})-E[\log(S_{T})]\label{eq:X}
\end{equation}
be the centralized log-returns. $E[\log(S_{T})]$ is the expectation of the log-returns under the risk-neutral measure and can be obtained from the characteristic function of $\log(S_{T})$. Because the density of $X$ is centered around zero, it is justified to approximate the density of $X$  by a \emph{symmetric} interval $[-L,L]$. In \cite{fang2009novel} a slightly different centralization of the log-returns has been considered. The characteristic function of $X$ is denoted by $\varphi_X$ and the density of $X$ is denoted by $f_X$.
In this setting, the price of a put option is given by
\begin{equation}
e^{-rT}E[(K-S_{T})^{+}]=e^{-rT}\int_{\mathbb{R}}v(x)f_{X}(x)dx,\label{eq:put}
\end{equation}
where
\begin{equation}
v(x)=\left(K-e^{x+E[\log(S_{T})]}\right)^{+},\quad x\in\mathbb{R}.\label{v}
\end{equation}
We approximate the integral at the right-hand-side of Equation (\ref{eq:put})
by the COS method. The cosine coefficients $v_{k}^{M}$ of $v$ can be
obtained in explicit form, see Appendix \ref{sec:Interval-cumulants}
or \cite{fang2009novel}. The price of a call option is given by the put-call parity. 

There are two formulae to choose the
truncation range $[-L,L]$ for the COS method: the \emph{cumulants
range}, based on Equation (\ref{eq:cumulants}) and the \emph{Markov range},
based on Corollary \ref{cor:bounded v}. To obtain the former, one needs to compute
the first, second and forth cumulants. Regarding the second, one need
to compute the $n^{th}-$moment. We will see that $n=4$, $n=6$ or
$n=8$ represent a reasonable choice. The $n^{th}-$moment and the $n^{th}-$cumulant
are similar concepts and can be obtained from the $n^{th}-$derivative
of the characteristic function. In Sections \ref{subsec:Black-Scholes-and-Laplace} till \ref{subsec:calibration}
we compare both formulae under the following aspects:
\begin{itemize}
\item How does the choice of the truncation range affect the number
of terms $N$? Certainly the larger the range, the larger we have
to set $N$ to obtain a certain precision.
\item Are there (important) examples where the COS method fails using the
Markov or the cumulants truncation range?
\item How can we avoid to evaluate the $n^{th}-$derivative of the characteristic
function to compute the $n^{th}-$moment or cumulant of the log-returns,
each time we apply the COS method, for performance optimization? 
\end{itemize}
In addition, we provide insights which moment to use for the Markov
range and we will apply the Markov range to all models discussed
in \cite{fang2009novel}. All numerical experiments are carried out on a modern
laptop (Intel i7-10750H) using the software R and vectorized code without parallelization.

\subsection{\label{subsec:Black-Scholes-and-Laplace}Black-Scholes and Laplace
model}

In the Black-Scholes model, see \cite{black1973th}, $X$ is normally
distributed. In the Laplace model, see \cite{madan2016adapted}, $X$
is Laplace distributed, see Equation (\ref{eq:Laplace density}) for the density of the Laplace
distribution. Both distributions are simple enough such that the quantile
functions are given explicitly. There are also closed-form solutions
for the prices of call and put options under both models.

Figure \ref{fig:Markov} compares $M$, defined
by the $n^{th}-$moment of $X$, see Equation (\ref{eq:M}), to the quantiles
of $X$. The higher $n$, the
sharper Markov's inequality. At least for the Normal and Laplace
distributions, good values for $M$ are obtained for $n\geq6$. Using
higher than the $8^{th}-$moment only marginally improves $M$. If not stated otherwise,
we will therefore use the $8^{th}-$moment to obtain the Markov truncation range for the COS method.

\begin{figure}[!h]
\centering
\includegraphics[height=7cm,width=7cm]{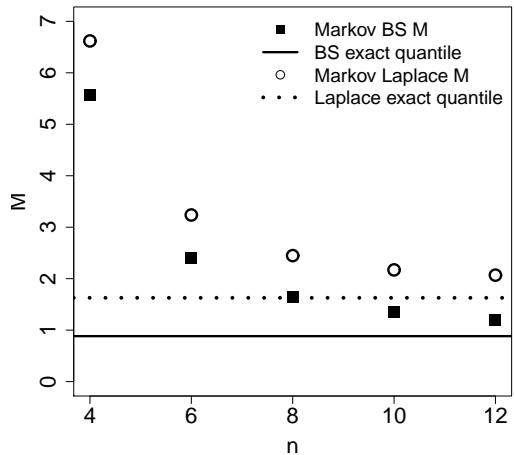}\\
\caption{\label{fig:Markov}Moments for Markov range. The variable $M$ to built the truncation range, see Equation (\ref{eq:M}), is shown for different moments and $K=1$ and $\varepsilon=10^{-5}$ for a Normal (BS) and Laplace distribution with mean zero and standard deviation $0.2$, respectively. The solid and dotted lines correspond to the quantiles, i.e. the minimal value for $M$ such that Equation (\ref{eq:vf<eps/2}) is satisfied (setting $v\equiv1$ in Equation (\ref{eq:vf<eps/2})).
}
\end{figure}

Consider an at-the-money call option on a stock with price $S_{0}=100$
today and with one year left to maturity. Assume the interest rates
are zero. The left panel of Figure \ref{fig:BS_Laplace} displays
$L$ over the error tolerance $\varepsilon$ for the Markov and the cumulants range.
The figure also shows the minimal value for $N$
such that the absolute difference of the true price of the option
and the approximation of the price using the COS method is below the
error tolerance $\varepsilon$. The the computational time to compute the option price by the COS method using $N$ steps is also indicated.

For a reasonable error tolerance,
e.g. $\varepsilon\in\left(10^{-3},10^{-8}\right)$, the minimal value
of $N$ to ensure the COS method is close enough to its reference
price is at most twice as large using the Markov range instead
of the cumulants range.

So using the Markov range based on the  $8^{th}-$moment instead of the well-established cumulants range,
at most doubles the computational time of the COS method for the same level
of accuracy. We see a similar pattern for advanced stock price models, see Table \ref{tab:Advanced-stock-price}.

The right panel of Figure \ref{fig:BS_Laplace} shows the effect of size of the truncation range, which is between one and a hundred times the volatility $\sigma=0.2$, on the computational time of the COS method. Setting $\varepsilon=10^{-3}$, we see a linear relationship between the size of the truncation range and the minimal value for $N$ to ensure convergence of the COS method for the Laplace model. The computational time is directly related to $N$. Setting the truncation range too large, increases the computational time unnecessarily. Thus Corollary \ref{cor:bounded v} may help to save computational time as well. For example, using the Markov range instead of $[-100\sigma,100\sigma]$, where $\sigma$ is the volatility of the log-returns, reduces the computational time by more than a factor two.

\begin{figure}[!h]
\centering
\begin{tabular}{cc}
\includegraphics[height=6.5cm,width=6cm]{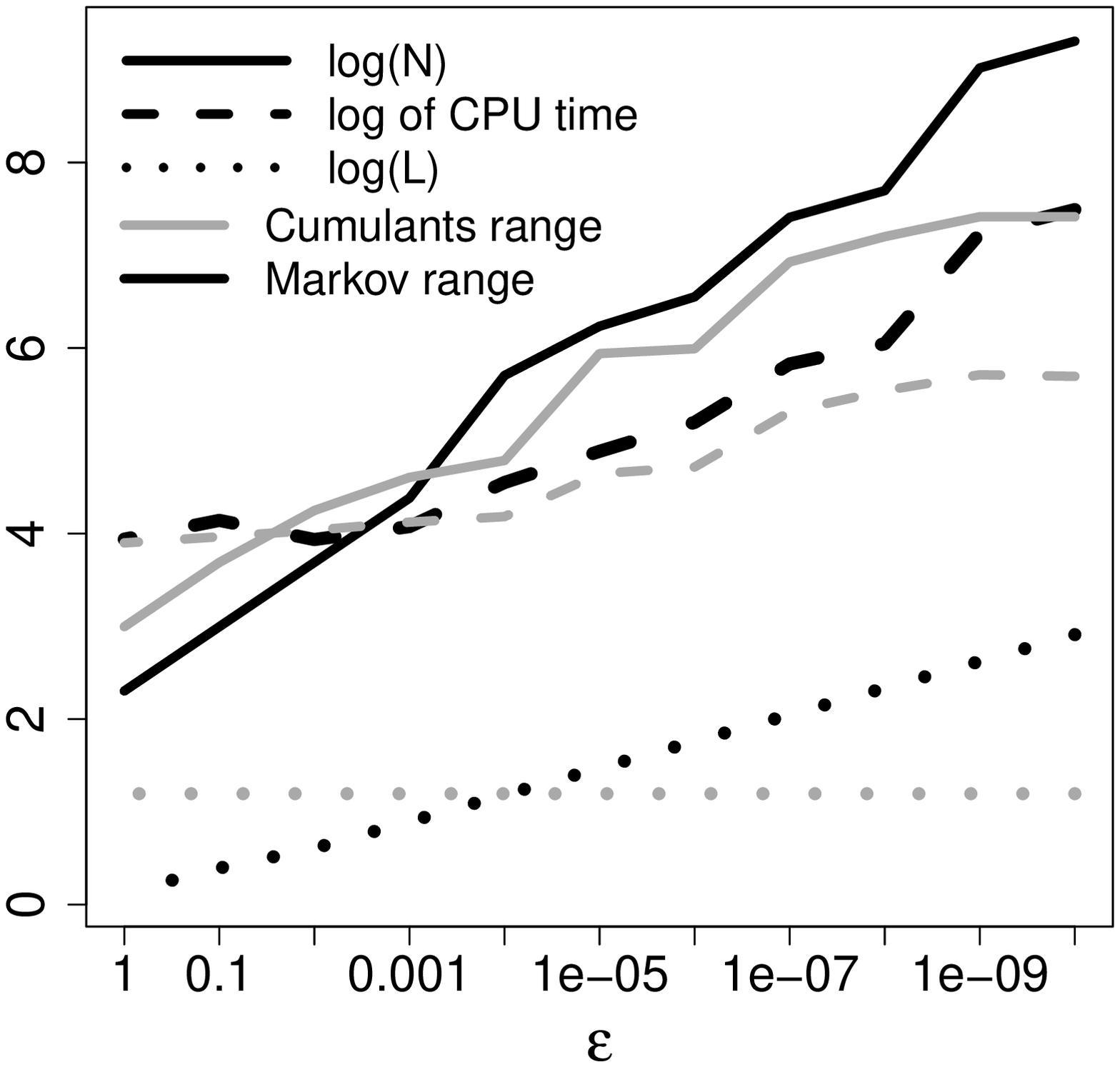}& \includegraphics[height=6.5cm,width=6cm]{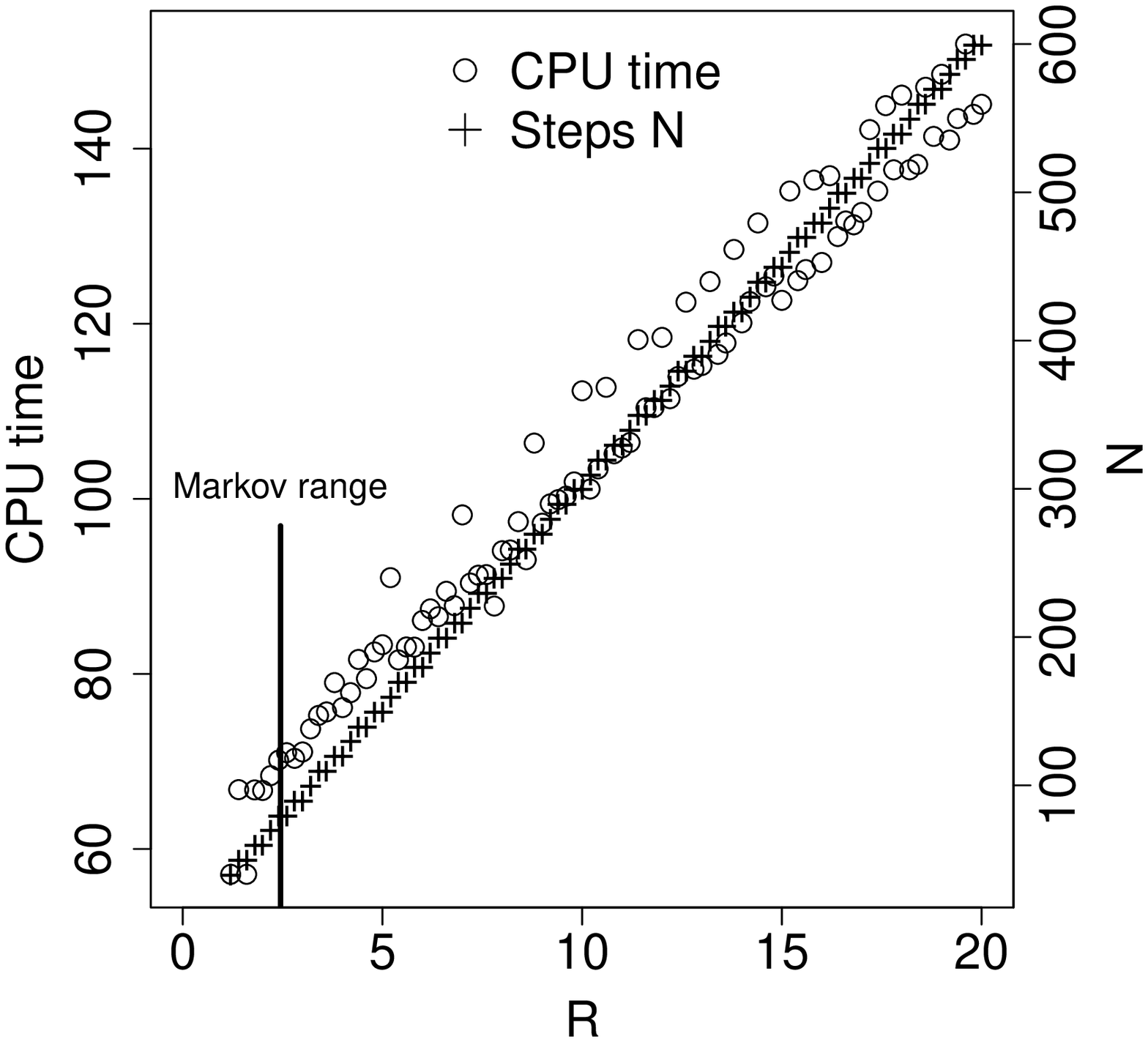}\\
Panel A&Panel B
\end{tabular}

\caption{\label{fig:BS_Laplace}
Convergence of the COS method for a call option in the Laplace model with volatility $0.2$. Panel A: The variable $L$ to built the truncation range is shown for the Markov range and the cumulants range over different $\varepsilon$. The minimal number of steps $N$ to obtain convergence of the COS method and the computational time (in microseconds) is illustrated as well. For the Markov range, the $8^{th}-$moment is used to obtain $L$, see Equation (\ref{eq:LL}). For the cumulants range four cumulants are used, see Equation (\ref{eq:cumulants}). Panel B: The minimal number of steps $N$ to obtain convergence of the COS method and the computational time (in microseconds) is illustrated for different truncation ranges $[-R,R]$ for $\varepsilon=10^{-3}$. }
\end{figure}

\subsection{\label{subsec:Advanced-models}Advanced models}

Fang and Oosterlee, see \cite{fang2009novel}, applied the COS method
with the cumulants range to three advanced stock price models with
different parameters, namely the Heston model, see \cite{heston1993closed},
the Variance Gamma model (VG), see \cite{madan1998variance}, and the
CGMY model, see \cite{carr2002fine}. We repeat the empirical study
using the Markov range instead. 

Table \ref{tab:Advanced-stock-price} shows the minimal value of $N$
to ensure the approximation of the price of an option by the COS method 
is close enough to its reference price, which
is taken from \cite{fang2009novel}, for both truncation ranges. For those
models, we conclude that using the Markov range based on the $8^{th}$
moment and an error tolerance of $\varepsilon=10^{-7}$ instead of
the cumulants range, increases $L$ by about the factor $2.5$
and $N$ by about the factor $2.2$. (Using the $4^{th}$ moment for
the Markov range and an error tolerance of $\varepsilon=10^{-4}$
increases $L$ and $N$ by about the factor four). The terms $N$ directly determine
the computational time of the COS method.

\begin{table}[h!]
	\begin{center}
		\caption{\label{tab:Advanced-stock-price}Advanced stock price models. Parameters
			for the first fourteen models are as in \cite{fang2009novel}. The parameters of the models M1, M2, M3 and M4 are specified in Section \ref{subsec:Counterexamples}. $\varepsilon$ describes the error tolerance and the columns $n_c$ and $n_M$,  describe the number of cumulants, respectively the number of moments, used to determine the truncation range. The columns $L_c$, $L_M$, $N_c$, $N_M$, 
			$t_c$, $t_M$ describe the truncation range, the minimal value of steps $N$ to ensure convergence of the COS method, and the computational time in microseconds, respectively for the cumulants range and the Markov range.}
		\begin{tabular}{llrrrrrrrrr}
			\parbox[t]{12mm}{\textbf{Model}} & \parbox[t]{12mm}{\textbf{Para-\\meters}} 
& \parbox[t]{6mm}{\centering \textbf{$\varepsilon$}} & \parbox[t]{6mm}{\centering \textbf{$n_c$}} & \parbox[t]{6mm}{\centering \textbf{$n_M$}}
			& \parbox[t]{6mm}{\centering \textbf{$L_c$}} & \parbox[t]{6mm}{\centering \textbf{$L_M$}}& \parbox[t]{6mm}{ \centering \textbf{$N_c$}}& \parbox[t]{6mm}{ \centering \textbf{$N_M$}} & \parbox[t]{6mm}{ \centering \textbf{$t_c$}}& \parbox[t]{6mm}{ \centering
\textbf{$t_M$}}\\[1.5\bigskipamount]
			\hline
            Heston & $T=1$ & $10^{-7}$ & 4 & 8 & 3.4 & 9.4 & 220 & 580 & 216 & 387\\
Heston & $T=1$ & $10^{-4}$ & 4 & 4 & 3.4 & 12 & 120 & 420 & 171 & 318\\
Heston & $T=10$ & $10^{-7}$ & 4 & 8 & 11.1 & 28 & 130 & 280 & 187 & 319\\
Heston & $T=10$ & $10^{-4}$ & 4 & 4 & 11.1 & 39.6 & 80 & 260 & 154 & 269\\
VG & $T=0.1$ & $10^{-7}$ & 4 & 8 & 0.8 & 2.3 & 630 & 950 & 242 & 326\\
VG & $T=0.1$ & $10^{-4}$ & 4 & 4 & 0.8 & 2.9 & 80 & 360 & 89 & 167\\
VG & $T=1$ & $10^{-7}$ & 4 & 8 & 1.9 & 4.3 & 100 & 190 & 80 & 103\\
VG & $T=1$ & $10^{-4}$ & 4 & 4 & 1.9 & 6.9 & 40 & 150 & 68 & 88\\
CMGY & $Y=0.5$ & $10^{-7}$ & 4 & 8 & 5.6 & 12.5 & 110 & 230 & 113 & 153\\
CMGY & $Y=0.5$ & $10^{-4}$ & 4 & 4 & 5.6 & 21.1 & 60 & 220 & 87 & 136\\
CGMY & $Y=1.5$ & $10^{-7}$ & 4 & 8 & 13.4 & 32.9 & 40 & 100 & 89 & 100\\
CGMY & $Y=1.5$ & $10^{-4}$ & 4 & 4 & 13.4 & 62.4 & 30 & 130 & 75 & 108\\
CGMY & $Y=1.98$ & $10^{-7}$ & 4 & 8 & 98 & 254.6 & 40 & 100 & 81 & 102\\
CGMY & $Y=1.98$ & $10^{-4}$ & 4 & 4 & 98 & 484.3 & 30 & 140 & 71 & 108\\
MJD & M1 & $10^{-7}$ & 4 & 8 & 0.9 & 4 & $-$ & 390 & $-$ & 164\\
MJD & M1 & $10^{-7}$ & 6 & 8 & 2.8 & 4 & 270 & 390 & 128 & 164\\
MJD & M2 & $10^{-8}$ & 6 & 8 & 5.8 & 18.2 & $-$ & 5750 & $-$ & 1444\\
CGMY & M3 & $10^{-7}$ & 4 & 8 & 1.5 & 9 & $-$ & 1990 & $-$ & 658\\
Heston & M4  & $10^{-2}$ & 2 & 8 & 1.3 & 3.7 & $-$ & 190 & $-$ & 211\\

		\end{tabular}
	\end{center}
\end{table}

\subsection{\label{subsec:Counterexamples}Examples where COS method
diverges}

In the Merton jump diffusion model (MJD), see \cite{merton1976option},
which is a generalization of the Black-Scholes model,
the stock price is modeled by a jump-diffusion process: the number of jumps are modeled by a Poisson process with intensity
$\eta>0$, i.e. the expected number of jumps in the time interval
$[0,T]$ is $\eta T$. The instantaneous variance of the returns,
conditional on no arrivals of jumps, is given by $\sigma^{2}>0$.
The jumps are log-normal distributed. The expected percentage jump-size
is described by $\kappa\in(-1,\infty)$. The variance of the logarithm
of the jumps is described by $\delta^{2}>0$. 

For model M1 we choose
\[
T=0.1,\quad\sigma=0.1,\quad\eta=0.001,\quad\kappa=-0.5,\quad\delta=0.2
\]
and for model M2, we set

\[
T=0.01,\quad\sigma=0.1,\quad\eta=0.00001,\quad\kappa=e^{-6.98}-1\approx-0.999,\quad\delta=0.2.
\]
For both models we set $S_{0}=100$ and $r=0$ and analyze a call option with strike $K=100$. 

Under the Merton jump diffusion model
the characteristic function and the density of the log-returns and
pricing formulae for put and call option are given in closed-form in terms of an infinite series. We use the first one hundred terms of that series to obtain reference prices for model M1 and M2 and we also apply the Carr-Madan formula, see \cite{carr1999option}, to confirm the reference prices. 

The left panel of Figure \ref{fig:Jump-diffusion} shows the reference price of a call option under model M1 and the prices using
the COS method with the truncation range $[-L,L]$ based on cumulants ($L=0.85$ using four cumulants)
and Markov's inequality ($L=3.99$ using the $8^{th}-$moment). An application of Corollary \ref{cor:bounded v} provides
a satisfactory result. However, we clearly see the approximation of the
price by the COS method does not converge properly using the cumulants
range. The relative error is about two basis points (BPS), a significant
difference, independent how large we choose $N$. The cumulants range
is too short and does not fully capture the second mode (the jump)
of the MJD density, as shown by the right panel of Figure \ref{fig:Jump-diffusion}.

Next we test the cumulants range with six cumulants for the MJD model. The truncation range using six cumulants for model M1 is large enough to ensure the convergence of the COS model, see Table \ref{tab:Advanced-stock-price}. But the cumulants range with six cumulants applied to model M2 is $[-5.8,5.8]$ and is \emph{not} large enough to ensure convergence within the required precision ($\varepsilon=10^{-8}$). Why? If a jump occurs, the expected jump size of the log-returns is equal to 
\[
\log(\kappa+1)-\frac{1}{2}\delta^{2}=-7,
\]
which is not inside the interval $[-5.8,5.8]$. Hence, again, the second mode of the density, i.e. the jump, is not fully captured by the truncation interval based on six cumulants. We report the prices for model M2 using different numeric approximations
\begin{align*}
\pi_{\text{MJD}} & =0.3989455935507185\\
\pi_{\text{Carr-Madan}} & =0.3989455935506932\\
\pi_{\text{Markov}} & =0.3989455935506925\\
\pi_{\text{Cumulants}} & =0.3989454898987361
\end{align*}

The price $\pi_{\text{MJD}}$ is obtained by the closed-form solution of the call price in terms of
an infinite series using the first one hundred terms. $\pi_{\text{Carr-Madan}}$
is obtained by the Carr-Madan formula where the damping-factor is
set to $\alpha=0.1$, we use $N=2^{17}$ points and the truncated
Fourier domain is set to $[0,1200]$. $\pi_{\text{Markov}}$ is obtained
by the COS method with $N=10^{6}$ terms and using the $8^{th}-$moment
and $\varepsilon=10^{-8}$ to obtain the truncation range. $\pi_{\text{Cumulants}}$ is obtained
by the COS method also with $N=10^{6}$ terms and using six cumulants
to obtain the truncation range.

Model M3 is a CGMY model with parameters $C=0.005$ and $G=M=Y=1.5$. Consider
an at-the-money call option on a stock with price $S_{0}=100$ today
and with $0.1$ years left to maturity. Assume the interest rates
are zero. Using the cumulative range with four cumulants, the relative error of the
approximation by the COS method and the reference price is about one
basis point, a significant difference, see Table \ref{tab:.-CGMY},
and does not improve when increasing $N$.

Model M4 is the Heston model with the following parameters: speed of mean-reversion
$\kappa=1$, level of mean-reversion $\eta=0.05$, vol of vol $\theta=2$,
initial vol $v_{0}=0.01$ and correlation $\rho=-0.75$. Consider
an at-the-money call option on a stock with price $S_{0}=100$ today
and with $0.5$ years left to maturity. Assume the interest rates
are zero. Set $\varepsilon=10^{-2}$. 

Using the cumulative range based only on the second cumulant, the truncation range is $[-1.33,1.33]$
and the price of the option by the COS method is $1.709$. 

On the other hand, $1.738$ is the price of the option based on the Markov
range, which is $[-3.71,3.71]$. We used $N=1000$. Using a larger
$N$ does not change the first three digits of the prices anymore.
We also applied the Carr-Madan formula to confirm the price $1.738$.

\begin{figure}[!h]
\begin{centering}
\begin{tabular}{cc}
\includegraphics[height=6.5cm,width=6cm]{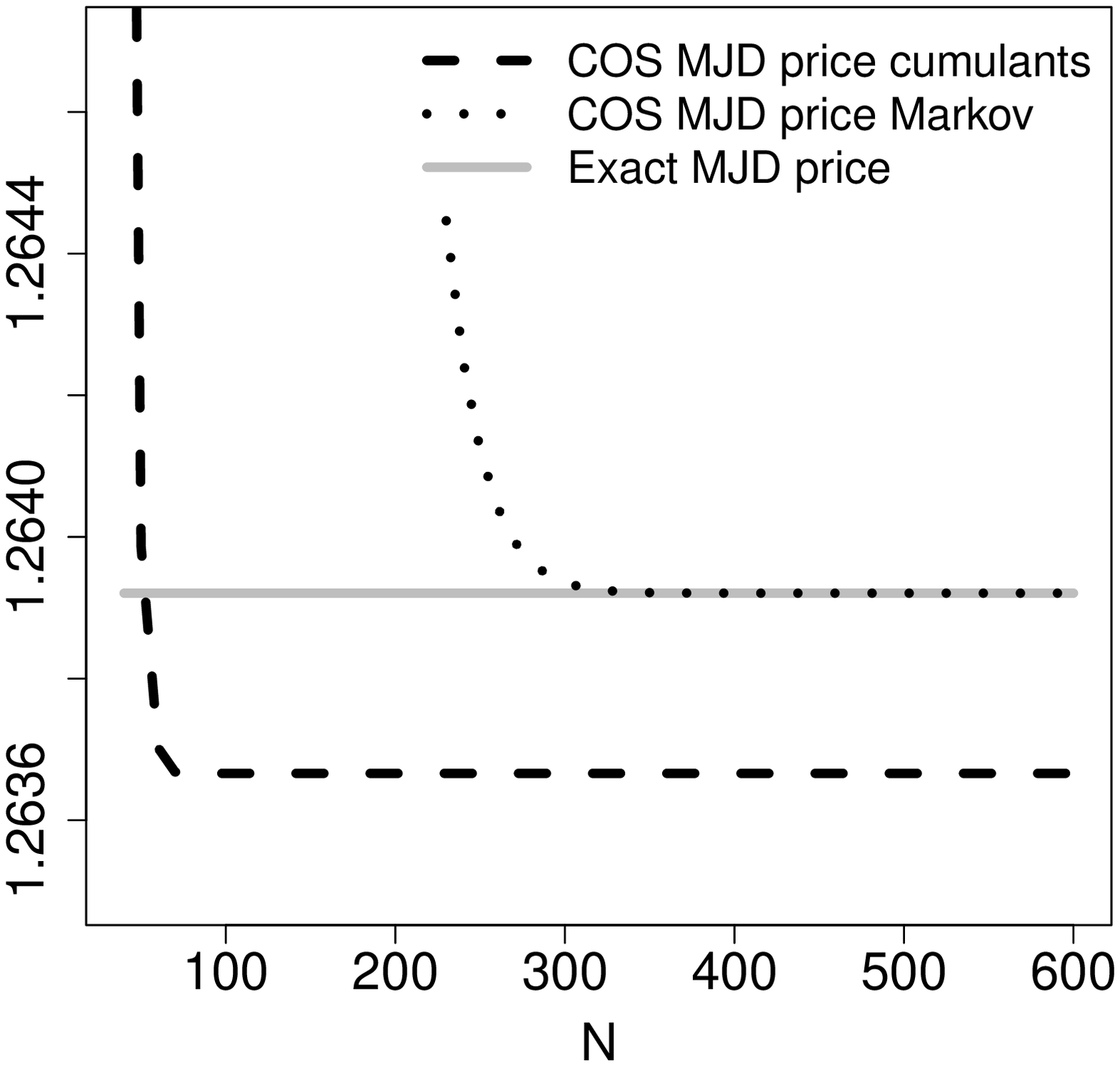}& \includegraphics[height=6.5cm,width=6cm]{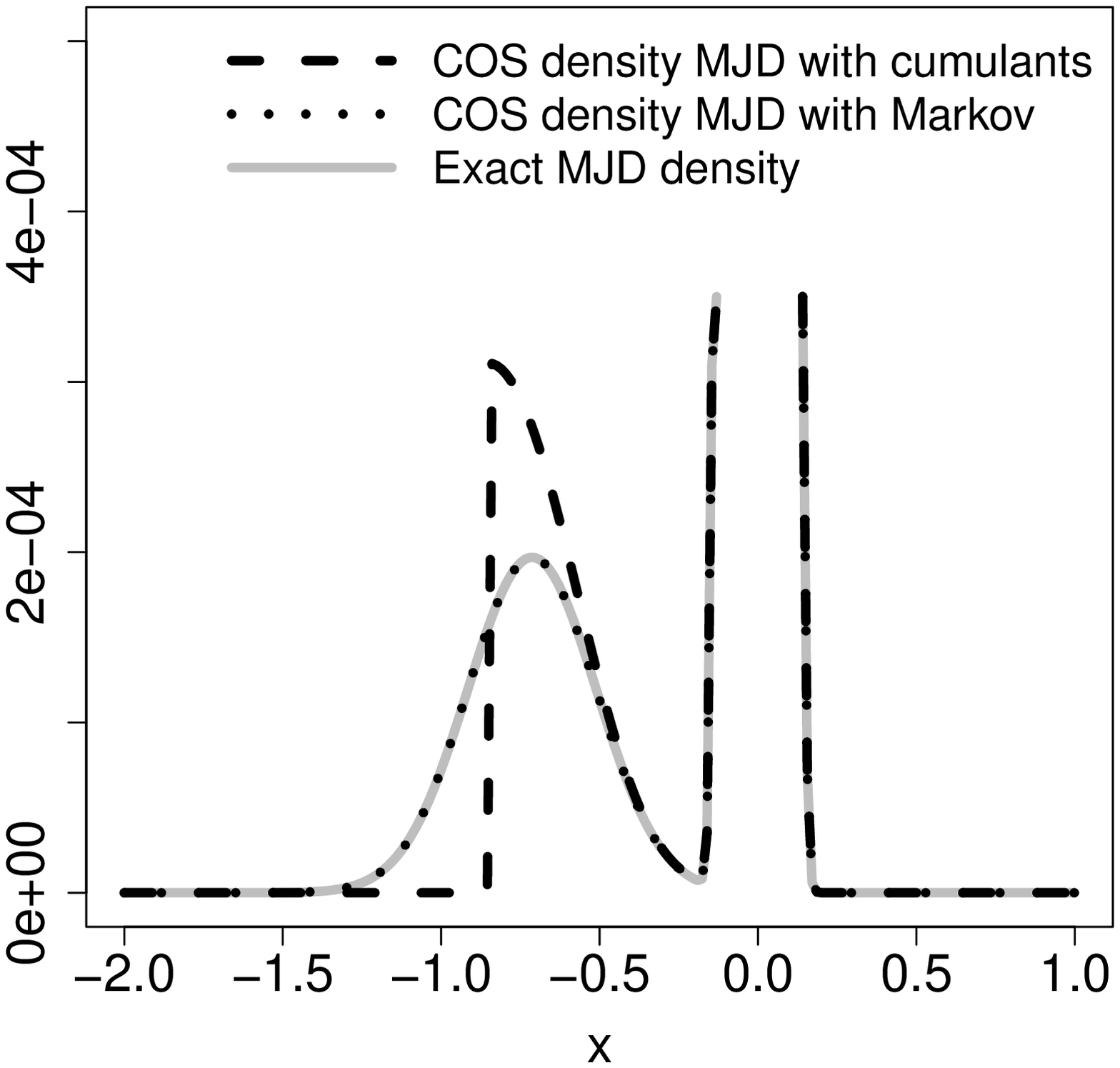}\\
Panel A&Panel B
\end{tabular}
\end{centering}
\caption{\label{fig:Jump-diffusion}COS method for MJD model. Panel A: convergence of COS call prices for model M1. The approximation of the call option price by the COS method, where the truncation range is based on four cumulants, is $1.263666$ but the reference price, i.e. the exact price, is $1.263921$. Panel B: MJD density and COS approximations. The density approximation by the COS method using four cumulants for the truncation range does not change no matter how large we choose $N$. We set $\varepsilon=10^{-7}$.}
\end{figure}

\begin{table}[h!]
  \begin{center}
    \caption{\label{tab:.-CGMY}CGMY model. Parameters $C=0.005$ and $G=M=Y=1.5$.
        Choose $\varepsilon=10^{-7}$. The reference price is 1.02168477497...,
        which we obtained using an approximating range ten times
        larger than the Markov range, i.e. $[-89,89]$, and $N=10^{7}$.
        Increasing the truncation range or $N$ does not change the first
        $10$ decimal digits of the reference price anymore. We set the
        difference between the reference price and the COS approximation to
        zero if the first $10$ decimal digits coincide.}
    \begin{tabular}{p{1cm}p{2cm}p{2cm}p{2.5cm}p{2.5cm}}
      \textbf{Terms \textit{N}} & \textbf{Markov abs. error} & \textbf{Cumul. abs. error} & \textbf{Markov rel. error in BPS} & \textbf{Cumul. rel. error in BPS} \\
      \hline
        $1000$ & $4.8\times10^{-4}$ & $1.07\times10^{-4}$ & $4.74$ & $1.04$\\
        $2000$ & $8.4\times10^{-8}$ & $1.07\times10^{-4}$ & $8.2\times10^{-4}$ & $1.04$\\
        $4000$ & $0$ & $1.07\times10^{-4}$ & $0$ & $1.04$\\
        $8000$ & $0$ & $1.07\times10^{-4}$ & $0$ & $1.04$\\

    \end{tabular}
  \end{center}
\end{table}

\subsection{\label{subsec:calibration}Application to model calibration}

Usually, one proceeds as follows to calibrate a stock price model,
like the Heston model, to real market data: given a set of market
prices of put and call options, minimize the mean square error between
market prices of the options and the corresponding prices predicted
by the model. During the optimization phase, model prices need to
be computed very often, e.g. by the COS method. 

We assume the model is described by $m$ parameters. Let $0<T_{0}\leq T_{1}$
be the smallest and largest maturity of the put and call options, respectively
let $\Theta\subset\mathbb{R}^{m}$ be the space of feasible parameters
of the model.

Let $X_{T}^{\theta}$ be the centralized log returns for the parameter
$\theta\in\Theta$ and maturity $T\in[T_{0},T_{1}]$. To compute the
price of a put or call option with strike $K$ by the COS method via
Corollary \ref{cor:bounded v}, we need to estimate the $n^{th}-$moment
\[
\mu_{n}^{\theta,T}:=E\left[\left(X_{T}^{\theta}\right)^{n}\right],
\]
to obtain an estimate for the truncation range $[-L,L]$ of the
density of $X_{T}^{\theta}$. The $n^{th}-$moment could directly
be determined by differentiating the characteristic function $X_{T}^{\theta}$
exactly $n-$times using a computer-algebra-system. 

In general, evaluating the $n^{th}-$derivative of a characteristic
function each time the COS method is called, might slow down the total
calibration time significantly because the $n^{th}-$derivative of
the characteristic function of some models can be very involved. For
fixed $n\in\mathbb{N}$, we therefore let
\begin{align*}
h:\Theta\times[T_{0},T_{1}] & \to[0,\infty)\\
(\theta,T) & \mapsto\mu_{n}^{\theta,T}.
\end{align*}
In the case of the Heston model, the function $h$ is continuous, see Lemma 1 in \cite{ruckdeschel2013pricing}. We propose to identify a function $\hat{h}$ as an approximation of $h$ upfront \emph{before} the calibration procedure. The evaluation of
$\hat{h}$ is expected to be fast. One might for example obtain $\mu_{n}^{\theta,T}$
for a (large) sample in $\Theta\times[T_{0},T_{1}]$ and defined $\hat{h}$
by a non-linear regression. 

Training $\hat{h}$ to the sample takes some
time but need to be done only once. The idea is to use $\hat{h}(\theta,T)$ as an approximation of $\mu_{n}^{\theta,T}$
to obtain the truncation range via Corollary \ref{cor:bounded v}. Even if $\hat{h}$
is only a rough estimate of the $n^{th}-$moment, we expect this approach
to work well because Markov's inequality usually overestimates the
tail sum, which provides us with a certain ``safety margin''.

We illustrate this idea for the Heston model. First, we define $\Theta$
for the Heston model, which has five parameters: the speed of mean
reversion $\kappa$, the mean level of variance $\eta$, the volatility
of volatility $\theta$, the initial volatility $v_{0}$ and the correlation
$\rho$. We assume
\[
\Theta\times[T_{0},T_{1}]=\left(10^{-3},10\right)\times\left(10^{-3},2\right)^{3}\times\left(-1,1\right)\times\left(\frac{1}{12},2\right).
\]

We randomly choose $5\times10^{5}$ values in $\Theta\times[T_{0},T_{1}]$
and compute $\mu_{8}^{\theta,T}$ for each of those values by a Monte Carlo simulation. (The $8^{th}-$derivative of the characteristic function of the Heston model is extremely involved). 

Then we train a small random forest, see \cite{breiman2001random},
consisting of $50$ decisions trees. We choose a random forest for
interpolation because the calibration is straightforward. We
are confident that other interpolation methods, e.g. based on neural
networks, produce similar results. 

Next, we define a test set and choose $1000$ values in $\Theta\times[T_{0},T_{1}]$
randomly. \cite{junikePerformance} calibrated the Heston model to a time series
of $100$ time points in Summer 2017 of real market data of put and
call options, including calm and more volatile trading days. We add
those $100$ parameters to our test set with a maturity of half a
year. 

For each parameter set, we compute a reference price\footnote{We computed the reference price by the COS method
using an approximating range two times larger than the Markov range and $N=10^{6}$ terms. We verified the reference price by the Carr-Madan formula, where the damping-factor is set to $\alpha=0.1$, we use $N=2^{17}$ points and the truncated Fourier domain is set to $[0, 1200]$.} of three
call options with strikes $K\in\{75,100,125\}$. We choose $S_{0}=100$, $r=0$ and $\varepsilon=10^{-4}$.

 We obtain very satisfactory results approximating prices of call options by the COS method if the truncation range is obtained via Corollary \ref{cor:bounded v}, where the $8^{th}-$moment 
is estimated by the random forest $\hat{h}$. The maximal absolute error for the market data test set over all options is less than $\varepsilon$ for $N\geq500$. The maximal absolute error for the random test set is less than $\varepsilon$ for $N\geq2000$.

Last we comment on the CPU time: computing one option price by the
COS method for $N=1000$ takes about $700$ microseconds using the software R and vectorized code
without parallelization. To evaluate $\hat{h}$ on $10^{4}$ parameter sets
using R's package \emph{randomForest,} also without parallelization,
takes on average $60$ microseconds per set. 

\section{\label{Conclusions}Conclusions}
The COS method is used to compute certain integrals appearing in mathematical ﬁnance by efficiently retrieving the probability
density function of a random variable describing some log-returns from its characteristic function.
The main idea is to approximate a density with infinite support on a finite range by a cosine series.

We provided a new framework to prove the convergence
of the COS method, which enables us to obtain an explicit formula for
the minimal length of the truncation range, given some maximal
error tolerance between the integral and its approximation by the COS
method. 

The formula for the truncation range is based on Markov's inequality and it is assumed that the density of the log-returns has semi-heavy tails. To obtain the truncation range, we need the $n^{th}-$moment of the (centralized) log-returns.
The larger $n$, the sharper Markov's inequality. From numerical experiments, we concluded that $n=8$, $n=6$ or even $n=4$ is a reasonable choice. 

The $n^{th}-$moment could directly be determined from the $n^{th}-$derivative of the characteristic function. In the case of the Heston model, the characteristic function is too involved and instead we successfully employed a machine learning approach to estimate the $8^{th}-$moment.

\section*{Acknowledgments}
We thank two anonymous referees for many valuable comments and suggestions enabling us to improve the quality of the paper.

\appendix 
\numberwithin{equation}{section}

\section{\label{sec:Interval-cumulants}Cosine-coefficients for a put option}
Define $X$ as in Equation (\ref{eq:X}) to compute the price of a put option with strike $K$ by the COS method using Markov's inequality for the truncation range. The characteristic function of $X$ is given by
\[
\varphi_{X}(u)=\varphi_{\log(S_{T})}(u)\exp\left(-iuE[\log(S_{T})]\right),
\]
where the expectation $E[\log(S_{T})]$ can be computed from the characteristic function of $S_T$ by
\[
E[\log(S_{T})]=-i\varphi_{\log(S_{T})}^{\prime}(0).
 \]
 Compute $n^{th}-$moment of $X$ with tricks from Section \ref{subsec:calibration} or by
\[
\mu_{n}=\frac{1}{i^{n}}\left.\frac{\partial^{n}}{\partial u^{n}}\varphi_{X}(u)\right|_{u=0},\quad n\in\{2,8\}.
\]
Choose $\sigma\geq\sqrt{\mu_{2}}$ and set $M$ and $L$ as in Equations (\ref{eq:M}) and (\ref{eq:LL}).
Let 
\begin{equation*}
d:=\min\left(\log\left(K\right)-E[\log(S_{T})],M\right).
\end{equation*}
If $d\geq-M$, compute the cosine-coefficients $v_{k}^{M}$ of $v$, defined in Equation (\ref{v}), by
\begin{align*}
v_{k}^{M}= & \int_{-M}^{M}\left(K-e^{x+E[\log(S_{T})]}\right)^{+}\cos\left(k\pi\frac{x+L}{2L}\right)dx\\
= & K\underbrace{\int_{-M}^{d}\cos\left(k\pi\frac{x+L}{2L}\right)dx}_{=:\Psi_{0}(k)}\\
 & -e^{E[\log(S_{T})]}\underbrace{\int_{-M}^{d}e^{x}\cos\left(k\pi\frac{x+L}{2L}\right)dx}_{=:\Psi_{1}(k)}.
\end{align*}
$\Psi_{0}$ and $\Psi_{1}$ can be computed easily:
\[
\Psi_{0}(k)=\begin{cases}
\frac{2L}{k\pi}\Big(\sin\big(k\pi\dfrac{d+L}{2L}\big)-\sin\big(k\pi\dfrac{-M+L}{2L}\big)\Big), & k>0\\
d+M, & k=0
\end{cases}
\]
and
\begin{align*}
\Psi_{1}(k) & =\left[e^{d}\left(\frac{k\pi}{2L}\sin\left(k\pi\frac{d+L}{2L}\right)+\cos\left(k\pi\frac{d+L}{2L}\right)\right)\right.\\
 & \left.-e^{-M}\left(\frac{k\pi}{2L}\sin\left(k\pi\frac{-M+L}{2L}\right)+\cos\left(k\pi\frac{-M+L}{2L}\right)\right)\right]\frac{1}{1+\left(\frac{k\pi}{2L}\right)^{2}}.
\end{align*}
Choose $N$ large enough and set 
\[
c_{k}^{L}:=\frac{1}{L}\Re\left\{ \varphi_{X}\left(\frac{k\pi}{2L}\right)e^{i\frac{k\pi}{2}}\right\} ,\quad k=0,1,...,N.
\]
Then it holds for the price of a European put option
\[
e^{-rT}E[(K-S_{T})^{+}]\approx e^{-rT}\sum_{k=0}^{N}{}^{\prime}c_{k}^{L}v_{k}^{M},
\]
where $\sum{}^{\prime}$ indicates that the first term in the summation
is weighted by one-half. The price of a call option can be computed using the put-call parity.


\begin{thebibliography}{99}
\bibitem{albin2009asymptotic}
J. Albin and M. Sund{\'e}n, On the asymptotic behaviour of L{\'e}vy processes, Part I: Subexponential and exponential processes.
Stoch Process Their Appl
{\bf  119(1)} (2009) 281--304.

\bibitem{bardgett2019inferring}
C. Bardgett, E. Gourier, M. Leippold, Inferring volatility dynamics and risk premia from the S\&P 500 and VIX markets.
J. Finance Econ.
{\bf  131} (2019) 593--618.

\bibitem{black1973th}
F. Black, M. Scholes, The pricing of options and corporate liabilities.
J. Polit. Econ. {\bf 81(3)} (1973) 637--654.

\bibitem{breiman2001random}
L. Breiman, Random forests.
Mach. Learn. {\bf 45(1)} (2001) 5--32.

\bibitem{carr2002fine}
P. P. Carr, H. Geman,  D. B. Madan, M. Yor, The fine structure of asset returns: An empirical investigation.
J. Bus. {\bf 75(2)} (2002) 305--332.

\bibitem{carr1999option}
P. P. Carr, D. P. Madan, Option valuation using the fast Fourier transform.
J. Comput. Financ {\bf 2(4)} (1999) 61--73.

\bibitem{dragulescu2002probability}
A. A. Dr\v{a}gulescu, V. M. Yakovenko, Probability distribution of returns in the Heston model with stochastic volatility.
Quant Finance
{\bf  2(6)} (2002) 443--453. 

\bibitem{fang2009novel}
F. Fang, C. W. Oosterlee, A novel pricing method for European
options based on Fourier-cosine series expansions. SIAM J. Sci. Comp.
{\bf  31} (2009) 826--848.

\bibitem{fang2009pricing}
F. Fang, C. W. Oosterlee, Pricing early-exercise and discrete
barrier options by Fourier-cosine series expansions. Num. Math. {\bf 114} (2009) 27--62.

\bibitem{fang2011fourier}
F. Fang, C. W. Oosterlee, A Fourier-based valuation method
for Bermudan and barrier options under Heston’s model.
SIAM J. Fin. Math. {\bf 2} (2011) 439--463.

\bibitem{heston1993closed}
S. L. . Heston, A closed-form solution for options with stochastic volatility with applications to bond and currency options.
Rev. Financ. Stud. {\bf 6(2)} (1993) 327--343.

\bibitem{grzelak2011heston}
L. A. Grzelak, C. W. Oosterlee,  On the Heston model with stochastic interest rates. SIAM J. Fin. Math. {\bf 2} (2011) 255--286.

\bibitem{hirsa2012computational}
A. Hirsa, Computational methods in finance. CRC Press, 2012.

\bibitem{horm}
L. H\"ormander, The analysis of linear partial differential operators I. Distribution theory and Fourier analysis.
2nd edition. Springer-Verlag, 1990. 

\bibitem{junikePerformance}
G. Junike, W. Schoutens, H. Stier, Performance of advanced stock price models when it becomes exotic: an empirical study.
Ann. Financ. (2021), https://doi.org/10.1007/s10436-021-00396-2.

\bibitem{leitao2018data} {\'A}. Leitao, C. W. Oosterlee, L. Ortiz-Gracia, S. M. Bohte,
On the data-driven COS method.
Appl. Math. Comput. {\bf 317} (2018) 68--84.

\bibitem{liu2019neural} S. Liu, A. Borovykh, L. A. Grzelak, C. W. Oosterlee,
A neural network-based framework for financial model calibration.
J. Math. Ind. {\bf 9(1)} (2019) 1--28.

\bibitem{liu2019pricing} S. Liu, C. W. Oosterlee, S. M. Bohte
Pricing options and computing implied volatilities using neural networks.
Risks {\bf 7(1)} (2019) 1--22.

\bibitem{madan1998variance}
D. B. Madan, P. P. Carr, E. C. Chang, The variance gamma process and option pricing.
Rev. Financ. {\bf 2(1)} (1998) 79--105.

\bibitem{madan2016adapted}
D. B. Madan, Adapted hedging.
Ann. Financ. {\bf 12(3)} (2016) 305--334.

\bibitem{merton1976option}
R. C. Merton, Option pricing when underlying stock returns are discontinuous.
J. Financ. Econ. {\bf 3(1-2)} (1976) 125--144.

\bibitem{oosterlee2019mathematical}
C. W. Oosterlee, L. A. Grzelak, Mathematical Modeling and Computation in Finance: With Exercises and Python and Matlab Computer Codes.
World Scientific, 2020.

\bibitem{Ortiz2013robust}
L. Ortiz-Gracia, C. W. Oosterlee, Robust pricing of European options with wavelets and the characteristic function, SIAM J. Sci. Comp. {\bf 35(5)} (2013) B1055-B1084.

\bibitem{Ortiz2016Shannon}
L. Ortiz-Gracia, C. W. Oosterlee, A highly efficient Shannon wavelet inverse Fourier technique for pricing European options, SIAM J. Sci. Comp. {\bf 38(1)} (2016) B118--B143. 

\bibitem{ruckdeschel2013pricing}
P. Ruckdeschel, T. Sayer, A. Szimayer, Pricing American options in the Heston model: a close look at incorporating correlation.
J. Deriv. {\bf 20(3)} (2013) 9--29.

\bibitem{ruijter2012two} M. J. Ruijter,  C. W. Oosterlee,
Two-dimensional Fourier cosine series expansion method for pricing financial options.
SIAM J. Sci. Comp. {\bf 34} (2012) B642--B671.

\bibitem{schoutens2003levy}
W. Schoutens, L{\'e}vy processes in finance: pricing financial derivatives.
Wiley, 2003. 

\bibitem{zhang2013efficient}
B. Zhang, C. W. Oosterlee, Efficient pricing of European-style
Asian options under exponential L\'evy processes based on Fourier cosine
expansions. SIAM J. Fin. Math. {\bf 4} (2013) 399--426.

\end{thebibliography}
\end{document}